\documentclass[12pt,a4paper,twoside]{article}
\usepackage{amsmath}
\usepackage{amsthm}
\usepackage{color}
\usepackage[psamsfonts]{amsfonts,amssymb}
\usepackage{hyperref}

\newtheorem{definition}{Definition}
\newtheorem{theorem}[definition]{Theorem}
\newtheorem{proposition}[definition]{Proposition}
\newtheorem{lemma}[definition]{Lemma}
\newtheorem{corollary}[definition]{Corollary}
\newtheorem{rem}[definition]{Remark}

\newtheorem{rems}[definition]{Remarks}

\newcommand\unit{\hbox{\rm 1\kern-2.8truept l}}

\newcommand\modd{\kern-8pt\mod}
\newcommand\modp{\kern-5pt\mod}

\newcommand{\tr}[1]{{\rm tr }(#1)}   

\title{Asymptotic equivalence of a subclass of WCLT non degenerate GKSL generators}

\begin{document}

\author{Jorge R. Bola\~nos--Servin$^1$}
\date{November 2023}

\author{Jorge R. Bola$\tilde{\textrm{n}}$os-Servin}

\maketitle

\abstract{
We prove the assymptotic equivalence of a sequence of block diagonal matrices with Toeplitz blocks. The blocks are the principal submatrices of an originating Toeplitz sequence with generating symbol of the Wiener class.
As an application, using the invariance of certain \textit{diagonal}  and \textit{cyclic-diagonal} operator subspaces of the GKSL generators of circulant and WCLT quantum Markov semigroups, the asymptotic equivalence of the families is proved under suitable hypothesis.}

\section{Introduction}
The keystone to the theory of the asymptotic study of eigenvalues and singular values of Toeplitz matrices is a well known result due to Szeg\"o.  The first Szeg\"o Limit Theorem describes the collective asymptotic distribution of the eigenvalues in the hermitian case while the stronger  Avram-Partner theorem is concerned with the collective asymptotic distribution of the singular values of (not necessarily hermitian) Toeplitz matrices. The starting point is always a sequence of the  $n\times n$ truncations of an infinite Toeplitz matrix whose elements are the Fourier coefficients of a function commonly referred by specialists as the generating \textit{symbol}. It was only recently that the interest in asymptotic formulas for individual eigenvalues sparked back interest although for  a few narrow classes of generating symbols\cite{BBGM2}. For instance in Maximenko\cite{BBGM} a quantile function is used to achieve uniform approximation of the singular values.
 
Circulant matrices are a subclass of Toeplitz matrices with very convenient properties, such as that form an abelian sub-algebra, therefore diagonal with respect to a common basis which allows explicit computation of their eigenvalues. These matrices arise in many contexts, from stochastic processes as infinitesimal generators of Markov chains over abelian finite groups to numerical  analysis as  optimal preconditioners for Toeplitz systems\cite{Chan}. The idea of deriving asymptotic results on Toeplitz matrices by comparison with conveniently (adhoc) constructed circulant matrices has been quite useful  and is standard now \cite{Gray,Strang,Pearl}.

Quantum Markov Semigroups (QMS) are weak$*$-continuous semigroups of completely positive identity preserving maps acting on a von Neumann algebra. They are widely used in applications to quantum physics to represent the memoryless evolution of an open quantum system. From a probabilistic point of view, they are a natural non commutative generalization of classical Markov semigroups. In the case when a quantum Markov semigroup is defined on a finite dimensional algebra, as considered in this work, it is always uniformly continuous.

The families of QMS considered in this paper are the circulant family (\cite{BQ}) and the weak coupling limit type (WCLT) family (\cite{WCLT}) . In spite of the latter being quite a rich class of semigroups, Fagnola and Quezada\cite{FQ} pointed out the disjointness of these families. However both types of QMS exhibit the invariance of certain operator subspaces: \textit{diagonal} subspaces $\mathcal V$ for WCLT and  \textit{cyclic-diagonal} subspaces $\mathcal B$ for circulant. For the circulant family, the invariance has been used, together with the nice features of circulant matrices,  to obtain the explicit computation of invariant states, the quantum entropy production rate, see \cite{BQ}, and some spectral properties including the spectral gap \cite{BC}. For the WCLT family on the other hand, although it was thought that the invariance alone was enough to characterize the non degenerate class of WCLT generators, it has been recently established in \cite{FQ} that this is not the case and additional hypothesis are needed.

The aforementioned families of operator subspaces $\mathcal V$ and $\mathcal B$ are closely related to the Toeplitz and circulant matrix structure, it is natural to ask whether an ad-hoc family of circulant QMS can be constructed in order to derive asymptotic results of a WCLT QMS. The purpouse of this paper is to obtain asymptotically equivalent sequences for completely positive (CP) operators and GKSL generators with a Toeplitz.

 The main results of asymptotic equivalence are contained in Theorems \ref{AE-CP}  and \ref{AE-WCLT}, and as a consequence results on the eigenvalue distribution follow in Corollary - -.   

The paper is organized as follows. In Section \ref{prel} we review general results on the asymptotic behaviour of  sequences of circulant and Toeplitz matrices and introduce the main objects of our study (Equation (\ref{1Rdecomp})  and Definition \ref{def}). We also develop a suitable vectorization process which will be helpful in the sequel. Section \ref{sec-toep} is devoted to the block diagonal matrix representations (with Toeplitz blocks) of  C.P. Toeplitz operators and the asymptotic equivalence to C.P. circulant operators is proved under the natural assumption that the generating sequence is in $\ell^1(\mathbb Z)$. In Section \ref{wclt-ae}, block diagonal matrix (with pseudo-Toeplitz blocks) representations of the subclass of WCLT generators with  C.P. Toeplitz part  are derived. The asymptotic equivalence of this subclass of generators and circulant generators is proved under the assumption of absolutely summability of the generating sequence.

\section{Preliminaries} \label{prel}
\subsection{Asymptotic behaviour of matrices}

The idea of asymptotic equivalence for sequences of matrices we use through out the paper, due to Gray et al.\cite{Gray}, provides a simplified and direct path to the basic eigenvalue distribution theorems. This method is implicit and  is a special case of more general results of Grenander and Szeg\"o\cite{Gren}.

For each $n \geq 1$ let $\mathbb C^n$ and $\mathbb C^{n^2}$ be endowed with the corresponding  Euclidean inner product. We consider the Hilbert spaces of continuous linear operators on each of the above spaces with the Hilbert-Schmidt inner product $\langle A, B \rangle_{HS}=\text{tr}(A^*B)$ respectively, namely, $\Big(\mathcal M_n(\mathbb C), \langle \cdot, \cdot \rangle_{HS} \Big)$ and
$\Big(\mathcal M_{n^2}(\mathbb C), \langle \cdot, \cdot \rangle_{HS} \Big)$.

Recall that the strong norm and the Hilbert-Schmidt norm on $\mathcal M_{m}(\mathbb C) $ are defined as
\begin{align*}\Vert A \Vert= \sup\left\{\frac{\langle u,A^*Au  \rangle^{\frac{1}{2}}}{\Vert u\Vert}: u\in \mathbb C^s \right\}, \ |A|=\langle A,A \rangle_{HS}^{\frac{1}{2}},\ A \in \mathcal M_{s}(\mathbb C),\ m=n,n^2. \end{align*}

The normalized Hilbert-Schmidt norm will be denoted by $|A|_m=\frac{1}{\sqrt{m}}|A|$.

In this note, a sequence of matrices $\{A_n\}_n$ is to be understood as a sequence increasing in size, i.e., $A_n \in \mathcal M_n(\mathbb C)$ for $n\geq 1$.
\begin{definition}Two sequences of matrices $\{A_n\}_n$, $\{B_n\}_n$ are said to be asymptotically equivalent if
\begin{enumerate} \item $A_n$, $B_n$ are uniformly bounded in strong norm: \begin{align*}\Vert A_n \Vert, \Vert B_n\Vert < M <\infty.\end{align*}
\item  $A_n - B_n$ goes to zero in the normalized Hilbert-Schmidt norm:
\begin{align*} \lim_{n\to \infty} |A_n-B_n|_n=0.\end{align*}  \end{enumerate}  
The asymptotic equivalence of two sequences $\{A_n\}_n$, $\{B_n\}_n$ is commonly denoted by $A_n \thicksim B_n$.\end{definition}

The asymptotic equivalence of matrices implies that product and inverses behave similarly. Although the asymptotic equivalence is not sufficient to specify the behaviour of individual eigenvalues, it does specify the average behaviour in the sense of its distribution. This is the fundamental theorem 	concerning asymptotic eigenvalue behavior.

\begin{theorem}\label{gray-theo} Let $\{A_n\}_n$, $\{B_n\}_n$ be asymptotically equivalent  sequences of matrices with eigenvalues $\alpha_{n,k}$ and $\beta_{n,k}$, respectively. Assume that the eigenvalues of either matrix converge, i.e., $$\lim_{n\to \infty}\frac{1}{n}\sum_{k=0}^{n-1}\alpha_{n,k}^s$$ exists and is finite for  any positive integer $s$.	 Then
\begin{align*}\lim_{n\to \infty} \frac{1}{n} \sum_{k=0}^{n-1}\alpha_{n,k}^s=\lim_{n\to \infty}\frac{1}{n}\sum_{k=0}^{n-1} \beta_{n,k}^s \end{align*}\end{theorem}

As a consequence of the above theorem, asymptotic equivalence implies that for any polynomial $p$ 
\begin{equation}\lim_{n\to \infty} \label{Hoff} \frac{1}{n} \sum_{k=0}^{n-1}p(\alpha_{n,k})=\lim_{n\to \infty}\frac{1}{n}\sum_{k=0}^{n-1} p(\beta_{n,k}), \end{equation}
even more, the Stone-Weirstrass theorem allows us to replace $p$ by an arbitrary function $F$ continuous  in $[-M,M]$ whenever the sequences are hermitian. In this case the sequences $\{\alpha_{n,k}\}_{k,n}, \{ \beta_{n,k}\}_{k,n}$ are said to be asymptotically equally distributed.
A stronger result was derived by Trench\cite{Trench} using the Wielandt-Hoffman theorem in which (\ref{Hoff}) is replaced by
\begin{align*}\lim_{n\to \infty} \frac{1}{n} \sum_{k=0}^{n-1}|F(\alpha_{n,k})-F(\beta_{n,k})|=0, \end{align*} where the eigenvalues are ordered noncreasingly, in this case the sequences are said to be asymptotically absolutely equally distributed.

\subsection{Circulant and Toeplitz asymptotic equivalence}

Recall that an $n \times n$ matrix $T_n=(t_{i,j})_{i,j}$  is said to be Toeplitz if $t_{i,j}=t_{i-j}$ and $C_n=(c_{i,j})_n$ is said to be circulant if $c_{i,j}=c_{j-i}$, where the subindex operation in the latter is $\mod n$. 

Following Gray in \cite{Gray}, one constructs a sequence of Toeplitz matrices $\{T_n\}_{n=1}^{\infty}$ (increasing in size) given a fixed collection $t=\{t_j\}_{j\in \mathbb Z}$, called the \textit{symbol}, satisfying some summability condition, for instance $t\in \ell^1(\mathbb Z)$. Then the sequence circulant matrices $\{\widetilde{C}_n\}_n$  is asymptotically equivalent to $\{T_n\}_{n=1}^{\infty}$ where
\begin{align}\label{gray} \tilde{c}_j=\left\{ \begin{array}{ll}t_0 & \text{ if } j=0, \\ t_{-j}+t_{n-j} & \text{ if } j=1,2,\ldots, n-1.\end{array}  \right.\end{align} 

%
	%
%
%
%

The aim of this paper is to obtain analogous results for completely positive (CP) linear operators with Toeplitz-like structure with respect to a fixed basis. This class of operators will be precisely defined in the sequel.

\subsection{Toeplitz and circulant CP maps and  GKSL generators}

Notice that any fixed orthonormal basis $\{e_i\}_i$ of $\mathbb C^n$ induces a \textit{diagonal} $\mathcal V=\{\mathcal V_l\}_{l=-(n-1)}^{n-1}$ and \textit{cyclic-diagonal} $\mathcal B=\{\mathcal B_j\}_{j=0}^{n-1}$ orthogonal  decomposition of $\mathcal M_n(\mathbb C)$ as follows.

For $j=0,1,\ldots,n-1,$  define the subspaces
\begin{align*}\mathcal V_{-j}&=\text{span}\{|e_i\rangle \langle e_{i+j}|:i=0,1,\ldots,n-1-j\}, \\
\mathcal V_{j}&	=\text{span}\{|e_{i+j}\rangle \langle e_i|:i=0,1,\ldots,n-1-j\},\\
 \mathcal B_j&=\text{span}\{|e_i\rangle \langle e_{i+j}|: i=0,1,\ldots,n-1 \}. \end{align*} 

The subindex operations in $\mathcal B_j$ are understood$\mod n$. The  $\mathcal V$-subspaces correspond to the slicing in diagonal,  subdiagonal and superdiagonal subspaces. The ordered sets  $\{l\in\mathbb Z: -(n-1)\leq l \leq n-1\}$ and $\{j\in \mathbb Z: j=0,1,\ldots,n-1\}$ will be referred as the \textit{diagonal} index set and \textit{cyclic-diagonal} index set respectively. 

 It is straightforward to verify that each family is orthogonal w.r.t. the Hilbert-Schmidt inner product, i.e.,
\begin{align} \label{1Rdecomp} \mathcal M_n(\mathbb C)=\overline{\bigoplus}_l\mathcal V_l  =\  \ \bigoplus_{j=0}^{n-1} \ \  \mathcal B_j.\end{align}
To simplify our notation we use the symbol $\overline{\bigoplus}_l$ to denote $\bigoplus_{l=-(n-1)}^{n-1} $, i.e., when the index set in the sum is the \textit{diagonal} index set.

Denote by $S$ the left shift operator (w.r.t. the orthonormal basis fixed above)  $$ Se_i=\left\{\begin{array}{lc} e_{i-1} & \mbox{ if $0<i\leq n-1$}, \\ 0 & \mbox{other,} \end{array} \right.$$ 
and $J$ be the \textit{cyclic} left shift operator $Je_i=e_{i-1}$. 
Subindex operations for $J$ are understood $\mod n.$

In recent years the invariance of $\mathcal V$ or $\mathcal B$ of GKSL generators has proven useful in the study of some QMS generators \cite{CF,BQ,BCQ}, this strategy can reduce the dimension in some cases from $n^2$ dimensional to an $n$ dimensional problem. 
In Theorem 3.2 and Corollary 3.1 of \cite{FQ}, the authors prove that the $\mathcal V$-invariance of a CP operator $\Phi$ induces the existence of a Kraus representation where each Kraus operator $V_l$ belongs to some subspace $\mathcal V_{m_l}$,  and thus it can be written either as 
\begin{align*}S^{m_l} M_{m_l} \ \text{ or }\  S^{*{m_l}} M_{m_l} \end{align*}
for some integer $m_l\geq 0$ and some multiplication operator $	M_{m_l}$. 

Similarly, the $\mathcal B$-invariance of $\Phi$ induces a Kraus representation in terms of the cyclic translation $J$. We state the proposition without proof since it consists of a simplified version of the original results pointed above.

\begin{proposition}Let $\Phi$ be a CP operator on $\mathcal M_n(\mathbb C)$ such that $\Phi(\mathcal B_l)\subseteq \mathcal B_l$ for all $l=0,1,\ldots,n-1$. Then there exists a Kraus representation $\Phi(x)=\sum_{l=0}^{n-1} L^*_l x L_l$ in which each $L_l$ can we written in one of the following forms:
\begin{align} \label{M-circ} J^mM_{m_l} \text{ or } J^{*m}M_{m_l} \end{align} for some integer $0 \leq m \leq n-1$ and some multiplication operator $M_{m_l}$. \end{proposition}

Throughout the rest of the paper we will exclusively work with the case when the multiplication operators above are multiples of the identity.

\begin{definition} \label{def} Let $\Phi:\mathcal M_n(\mathbb C)\longrightarrow \mathcal M_n(\mathbb C)$ be a completely positive linear operator. \begin{enumerate}
\item  $\Phi$ is said to be a CP Toeplitz operator (w.r.t. to the basis $\{e_i\}_i$)  if it has a representation of the form 
\begin{align} \label{CP-toep} \Phi(x) = \sum_{j=0}^{n-1} t_{j} S^{*j} x S^j+ \sum_{j=1}^{n-1} t_{-j} S^j x S^{*j}, \ \ \ \ t_j \geq 0.  \end{align}

\item  $\Phi$ is said to be a CP circulant operator (w.r.t. to the basis $\{e_i\}_i$) if it has a representation of the form 
\begin{align} \Phi(x) = \sum_{j=0}^{n-1} c_{n-j} J^{*j} x J^j, \ \ \ \ c_j\geq 0.\end{align}\end{enumerate}
\end{definition}

Recall that a quantum Markov semigroup on a finite dimensional algebra is always uniformly continuous and its infinitesimal generator $ \mathcal L$ is a bounded operator characterized by a particular form, usually called canonical Gorini-Kossakowski-Sudarshan-Lindblad (GKSL) form:
\[
\mathcal L(x)=\Phi(x) +G^*x+xG\]
where $\Phi$ is a CP operator and $G=-\frac{1}{2}\Phi(\unit_n) -iH$ with $H=H^*\in \mathcal M_n(\mathbb C)$. This representation is not unique. We shall consider GKSL generators with representations such that the CP part is Toeplitz or Circulant as follows.

\begin{definition}\label{GKSL-def}
Let   $\mathcal L$ be the GKSL generator of a uniformly continuous QMS.
\begin{enumerate}
\item  $\mathcal L$ is said to be a Toeplitz QMS generator if the diagonal subspaces $\mathcal V$ are invariant for $\mathcal L$ and there exists a representation such that its CP (i.e., dissipative) part is Toeplitz.

\item  $\mathcal L$ is said to be a circulant QMS generator if the diagonal subspaces $\mathcal B$ and there exists a representation such that its CP (i.e., dissipative) part is circulant.
\end{enumerate}
\end{definition}

The class of Circulant QMS has been extensively studied in \cite{BC,BQ,BCQ}. 

\begin{rem} The $\mathcal V$-invariance ($\mathcal B$-invariance) alone is not enough to have CP Toeplitz (circulant) operators. The additional condition on the multiplication operators  is needed to have a CP Toeplitz (circulant) behaviour.
\end{rem}

\subsection{Toeplitz and circulant induced vectorizations}

Vectorization expresses through coordinates an isomorphism $\mathbb C^{n} \otimes \mathbb C^n\cong \mathbb C^{n^2}$ which is a unitary transformation compatible with the Euclidian inner product of $\mathbb C^{n^2}$.

The most common vectorization used in the literature consists in stacking  each column of $x=(x_{i,j})_{ij}\in \mathcal M_n(\mathbb C)$ into an $n^2$ coordinate vector
\begin{equation*}\text{vec}(x)=(x_{0,0},x_{1,0},\ldots,x_{n-1,0},x_{0,1},\ldots,x_{n-1,1},\ldots,x_{0,n-1},\ldots,x_{n-1,n-1})^T \in \mathbb C^{n^2}.\end{equation*}

We will use two different vectorization isomorphisms $v$ and $b$ with the difference being the order in which the elements are arranged in a vector. This order is  directly induced by the structure of the $\mathcal V$ and $\mathcal B$ subspaces, respectively.

\noindent For  each $-(n-1) \leq l \leq n-1$ let $\{e_i^l\}_i$ be the canonical basis of $\mathbb C^{n-|l|}$.

Define $  v_{l}: \mathcal V_{l} \longrightarrow\mathbb C^{n-|l|}$ as

$v_l:\left\{\begin{array}{c}|e_i \rangle \langle e_{i-l}| \longmapsto e_i^{l}  \ \ \text{ if } l<0, \\ |e_{i+l} \rangle \langle e_i| \longmapsto e_i^{l} \ \   \text{  if }\ l\geq 0, \end{array}\right.$

and for  each $0 \leq j \leq n-1$ define $ b_{j}:\mathcal B_n \longrightarrow \mathbb C^n$ as $|e_i \rangle \langle e_{i+j}|\longmapsto e_i.$
 
 Each $v_l$ and $b_{j}$ is an isometry. 

Let $v$ and $b$ be the isomorphisms between $\mathcal M_n(\mathbb C)$ and $\mathbb C^{n^2}$ induced by the smaller isomorphisms $\{v_l\}_l$, $\{b_{j}\}_j$ respectively, so
\begin{align*}v(x)&=v_0(x_0)\oplus \bigoplus_{l=1}^{n-1} v_{-l}(x_{-l})\oplus v_l(x_l)\ \ \ \ \text{ and }\ \ \ 
b(x)=b_0(\mathrm x_0)\bigoplus_{j=1}^{n-1} b_j(\mathrm x_j),\end{align*}
where $x=\bigoplus_lx_l=\bigoplus_j\mathrm{x}_j,$ $ x_l\in \mathcal V_l, \mathrm{x}_j \in \mathcal B_j,$  are the unique orthogonal decomposition of $x$ w.r.t. to $\mathcal V$ and $\mathcal B$ respectively.

Thus $v(x)$ is the vectorization of $x$ which rearranges its elements  w.r.t. to the \textit{diagonal} subspaces in the \textit{diagonal ordering} and $b(x)$ is the  one that rearranges them w.r.t. to the \textit{cyclic-diagonal} subspaces. 

It is straightforward to verify that the vectorizations $v$ and $b$ are compatible with the Euclidian (HS) inner product of $\mathbb C^{n^2}$ as well,
$$\langle x , y \rangle_{HS} =\tr{x^*y}= \langle v(x),v(y) \rangle_{\mathbb C^{n^2}}=\langle b(x),b(y) \rangle_{\mathbb C^{n^2}} \ \text{ for } x,y \in \mathcal M_n(\mathbb C).$$

We denote with $\Phi \cong A$
whenever a matrix $A\in \mathcal M_{n^2}(\mathbb C)$ represents a linear operator $\Phi$ acting on $\mathcal M_n(\mathbb C)$ w.r.t. the reordering of the canonical basis $\{|e_i\rangle \langle e_j |\}_{i,j}$ induced by the \textit{diagonal} or \textit{cyclic-diagonal} index sets.	

\subsection{Circulant C.P. maps}\label{circ-prop}

Circulant CP operators appear as the dissipative part of the GKSL generators of circulant Quantum Markov Semigroups (2. in Def. \ref{GKSL-def}) in \cite{BQ}. Some of the spectral properties of these generators have been studied in  \cite{BC} together with estimations for the spectral gap. 
Every CP circulant operator $\tilde{\Phi}$ (Def. \ref{def}) has an associated circulant matrix $C^{(n)}=(c_{j-i})_{ij}$. The properties satisfied by a circulant GKSL generator and its associated circulant matrix can be found in \cite{BQ}. We recall the main property we shall use in the sequel.

\begin{proposition} \label{iso-C} A circulant CP map $\tilde{\Phi}$ with asociated circulant matrix $C^{(n)}$ satisfies
\[\tilde{\Phi}\cong \bigoplus_{j=0}^{n-1}  C^{(n)} = \unit_n \otimes C^{(n)},\]
therefore  $b\left(\tilde{\Phi}(x)\right)=( \unit_n \otimes C^{(n)} ) b(x)$ for all $x\in \mathcal M_n(\mathbb C).$

 As a consequence  $ \Vert \tilde{\Phi} \Vert= \Vert  \unit_n \otimes C^{(n)}  \Vert$ and $ |\tilde{\Phi}|_{n^2}=| \unit_n \otimes C^{(n)}   |_{n^2}.$

\end{proposition}

\section{CP Toeplitz maps} \label{sec-toep}

In this section we shall provide analogous results for CP Toeplitz maps in terms of an associated Toeplitz matrix and the subspaces $\mathcal V$.

Let $\Phi$  be a CP Toeplitz map (Def.\ref{def}, \textit{1.}) and let $T^{(n)}_0=(t_{i-j})_{i,j}$ be  the $n \times n$  Toeplitz matrix whose entries are the scalars from the constant multiplication operators in $\Phi$. The following lemma deals with the restriction of $\Phi$ to the $\mathcal V$-subspaces and its relation to some submatrices of $T_0^{(n)}$. It is the analogous version of Proposition \ref{iso-C} to the case of Toeplitz CP maps.

\begin{lemma}  Let $\Phi$ be a C.P. Toeplitz operator. 
 The restriction of $\Phi$ to the subspace $\mathcal V_l$ is given by the $n-|l| \times n-|l|$ leading principal  submatrix $T^{(n)}_{|l|}$ of $T^{(n)}_0$,
\[T^{(n)}_{|l|}=\sum_{\substack{i,k=0}}^{n-1-|l|} \langle e_i, T^{(n)}_0 e_k \rangle	 |e_i\rangle \langle e_k|, \] that is, for $j=0,1,\ldots,n-1,$
\begin{align*} &v_{-j}\left(\Phi(|e_i \rangle \langle e_{i+j}|)\right)= T^{(n)}_{j}e_i \ \ i=0,1,\ldots,n-1-j,\\
&v_j\left(\Phi(|e_{i+j} \rangle \langle e_i|)\right)= T^{(n)}_{j}e_i \ \ i=0,1,\ldots,n-1-j. \end{align*} Therefore by linearity
\begin{eqnarray}\label{phi_restriction} v_l(\Phi|_{\mathcal V_l}(x))=	T^{(n)}_{|l|} v_l(x)\  \text{ for any} \ x\in\mathcal M_n(\mathbb C) \text{ and } l=-(n-1),\ldots,0,\ldots,n-1.\end{eqnarray}
\end{lemma}

\begin{proof} 
Let $j	\geq 0$ and let $|e_i \rangle \langle e_{i+j}|$ be any of the basic elements of $\mathcal V_{-j}$. Then, explicit computations using $S$ and $S^*$ yield
\begin{align*}  \Phi(|e_i \rangle \langle e_{i+j}|)&=\sum_{\substack{r=0\\0\leq i+r\leq n-1-j} }^{n-1} t_r |e_{i+r}\rangle \langle e_{i+r+j}|+\sum_{\substack{r=1\\ 0 \leq i-r\leq n-1-j}}^{n-1} t_{-r}|e_{i-r} \rangle \langle e_{i-r+j}|\\
&=\sum_{\substack{r=0\\0\leq i+r \leq n-1-j}}^{n-1}t_r  v_{-j}^{-1}(e^j_{i+r})+\sum_{\substack{r=0\\0\leq i-r\leq n-1-j}}^{n-1}t_{-r} v_{-j}^{-1}( e^j_{i-r})  \\
&= v_{-j}^{-1}\left( \sum_{\substack{r=0\\0\leq i+r \leq n-1-j}}^{n-1-j}t_r  e^j_{i+r}+\sum_{\substack{r=0\\0\leq i-r\leq n-1-j}}t_{-r}  e^j_{i-r}  \right)=v_{-j}^{-1}(T_j e^j_i), \end{align*}

where the fact that  $v_{-l}(|e_k\rangle \langle e_{k+j}|)=e_k$ was used in the second row and for the last equality follows since
\begin{align*} T^{(n)}_{j}=\sum_{r=0}^{n-1-j}  \left( \sum_{\substack{k=0\\0\leq k+r \leq n-1-j}}^{n-1-j} t_{-r}| e^j_k \rangle \langle e^j_{k+r}| + t_r  | e^j_{k+r} \rangle \langle e^j_k|  \right). \end{align*}
This proves (\ref{phi_restriction}) for $\mathcal V_{-j}$. The remaining case is analogous.

\end{proof}

In virtue of the following theorem, we will say $T^{(n)}_0$ is the associated Toeplitz matrix of $\Phi$. 

\begin{theorem} \label{iso-T} A Toeplitz CP map $\Phi$ with associated Toeplitz matrix $T^{(n)}_0$ satisfies

\begin{align*} \Phi\cong \displaystyle \overline{\bigoplus_l} T^{(n)}_l \ \ \ \text{ and }\ \ \ v\left(\Phi(x)\right)=\left(\overline{\bigoplus_l} T^{(n)}_l\right) v(x)\  \text{ for all } x\in \mathcal M_n(\mathbb C)\end{align*}
w.r.t. the diagonal ordering.

As a consequence the norms satisfy $ \displaystyle\Vert \Phi \Vert= \left\Vert \overline{\bigoplus_l} T^{(n)}_l \right\Vert$ and $\big|\Phi\big|_{n^2}=\displaystyle \left|\overline{\bigoplus_l}	 T^{(n)}_l \right|_{n^2}.$

\end{theorem}

\begin{proof}
The first part follows from the last proposition. For the strong norm it is straightforward to see that $\displaystyle \Big\Vert \Phi(x)\Big\Vert^2 = \langle \Phi(x),\Phi(x) \rangle_{HS} = \Big \langle v(\Phi(x)),v(\Phi(x)) \Big\rangle_{\mathbb C^{n^2}} $

 $ \ \  \ \ \ \  \ \ \  \ \ \ \  \ \ \ \ \ \ \ \  \ \ \ \  \ \ \ \ \ \displaystyle=\Big\langle \overline{\bigoplus_l} T^{(n)}_l v(x),\overline{\bigoplus_l} T^{(n)}_l v(x) \Big \rangle_{\mathbb C^{n^2}} = \Big \Vert \overline{\bigoplus_l} T^{(n)}_l v(x) \Big \Vert ^2 .$
 The desired result follows since $\Vert x \Vert=\Vert v(x) \Vert$.

For the HS norm, let  $e_{ij}=|e_i\rangle \langle e_j|$, then since $v$ is unitary,
\begin{align*}|\Phi|^2_{n^2}&=\frac{1}{n^2}\sum_{i,j} \Big\langle \Phi(e_{ij}),\Phi(e_{ij}) \Big\rangle_{HS}=\frac{1}{n^2}\sum_{i,j} \Big\langle v(\Phi(e_{ij})),v(\Phi(e_{ij})) \Big\rangle_{\mathbb C^{n^2}}  \\&=\frac{1}{n^2} \sum_{i,j} \Big\langle \overline{\bigoplus_l} T^{(n)}_l v(e_{ij}), \overline{\bigoplus_l} T^{(n)}_l v(e_{ij}) \Big\rangle_{\mathbb C^{n^2}} = \Big|\overline{\bigoplus_l} T^{(n)}_l\Big|_{n^2}^2.\end{align*} Again, reordering the basis the conclusion follows.

\end{proof}

The last theorem shows that with the proper reordering of the basis, any Toeplitz CP operator $\Phi$ on $\mathcal M_n(\mathbb C)$ is represented by $\overline{\bigoplus_l} T^{(n)}_l$, a $n^2 \times n^2$ block diagonal matrix  with $n-|j| \times n-|j|$ Toeplitz blocks. Each block is a leading principal submatrix of $T^{(n)}_0$ and $T^{(n)}_j=T^{(n)}_{-j}$.

\begin{rem} $\overline{\bigoplus_l} T^{(n)}_l$ is a block diagonal matrix with Toeplitz blocks but it is not a Toeplitz $n^2 \times n^2$  matrix. As far as we know this Toeplitz-like matrix structure has not been studied before.\end{rem}

Denoting the spectrum of $A$ by $\sigma(A)$, an immediate corollary follows.
\begin{corollary}The following properties hold:
\begin{itemize}
\item[a)]
\begin{equation}\label{block-spectrum} \sigma(\Phi)=\bigcup_{j=0}^{n-1} \sigma(T^{(n)}_j).\end{equation} 
\item[b)] Each $\lambda \in \{t_0, \ t_0 -\sqrt{t_1 t_{-1}},\  t_0 +\sqrt{t_1 t_{-1}} \}$ is always an eigenvalue of $\Phi$ with associated eigenvectors given by
\begin{align*}\lambda= t_0 -\sqrt{t_1 t_{-1}} &\longrightarrow \left\{\begin{array}{c} -\sqrt{t_{-1}}|e_{n-2}\rangle \langle e_0|+\sqrt{t_1}|e_{n-1}\rangle \langle e_1| \\ -\sqrt{t_{-1}}|e_0\rangle \langle e_{n-2}|+\sqrt{t_1}|e_1\rangle \langle e_{n-1}| \end{array} \right. \\ \lambda= t_0 +\sqrt{t_1 t_{-1}} &\longrightarrow \left\{\begin{array}{c} \ \  \sqrt{t_{-1}}|e_{n-2}\rangle \langle e_0|+\sqrt{t_1}|e_{n-1}\rangle \langle e_1| \\ \ \  \sqrt{t_{-1}}|e_0\rangle \langle e_{n-2}|+\sqrt{t_1}|e_1\rangle \langle e_{n-1}|. \end{array}\right.  \\
\lambda=t_0 \ \ \ \ \ \ \ \ \ \ \ \ \ &\longrightarrow \left\{ \begin{array}{c}|e_{0}\rangle \langle e_{n-1}| \\ |e_{n-1}\rangle \langle e_0|  \\  \end{array} \right.   \end{align*}

\end{itemize}\end{corollary}

\section{Asymptotic equivalence of CP Toeplitz and Circulant maps}
Let $\{t_l\}_{l=-\infty}^\infty$  be a non-negative sequence of the Wieiner class and for each $n\geq 1$ let $\Phi^{(n)}:\mathcal M_n(\mathbb C)\longrightarrow \mathcal M_n(\mathbb C)$ be the CP Toeplitz operator 
\begin{align*} \Phi^{(n)}(x) = \sum_{j=0}^{n-1} t_{j} S^{(n)*j} x S^{(n)^j}+ \sum_{j=1}^{n-1} t_{-j} S^{(n)^j} x S^{(n)*j},  \end{align*}
thus $\{\Phi^{(n)}\}_n$ is a sequence of CP Toeplitz maps associated with $t$.

Notice that $\{T_0^{(n)}\}_n$ is both the sequence of Toeplitz matrices associated to $\{\Phi^{(n)}\}_n$ and the Toeplitz sequence associated to the symbol $f(\eta)=\sum_k t_k e^{ik\eta}, \eta \in [0,2\pi]$.

Our purpose is to find a sequence of CP circulant maps $\{\widetilde{\Phi}^{(n)}\}_n$  asymptotically equivalent to  $\{\Phi^{(n)}\}_n$. It is well known that if $\{C^{(n)}\}_n$ is the sequence of circulant matrices given by 
$C^{(n)}=(c_{j-i}^{(n)})_{ij}$ where the subindex operations are mod $n$ and
\begin{align}\label{circ-matrix}c^{(n)}_j=\left\{ \begin{array}{ll}t_0 & \text{ if } j=0, \\ t_{-j}+t_{n-j} & \text{ if } j=1,2,\ldots, n-1.\end{array}  \right.\end{align} then  $T_0^{(n)} \thicksim C^{(n)}$, see \cite{Gray}.

Let $\{\widetilde{\Phi}^{(n)}\}_n$ be the sequence of  circulant CP maps such that their associated circulant matrices are precisely $\{C^{(n)}\}_n$, that is,
\begin{align}\label{CP-circ}  \widetilde{\Phi}^{(n)}(x)=\sum_{j=0}^{n-1} c_{n-j}^{(n)}J^{(n)*j}x J^{(n)j}, \ \ \ n\geq 1.\end{align}

The following standard results will be useful to manage the norms of the blocks involved.

\begin{lemma}Let $\mathcal H_1, \mathcal H_2$ be  finite dimensional Hilbert spaces and  $A_i \in \mathcal B(\mathcal H_i),$ $i=1,2$. Then $\Vert A_1\oplus A_2\Vert=\left\Vert \begin{pmatrix} A & 0 \\ 0 & B \end{pmatrix} \right\Vert=\max\{\Vert A\Vert ,\Vert B\Vert \}$ . \end{lemma}

 \begin{lemma} For each $n \geq 1$ let $T_j^{(n)}$ be the principal  $\ n-j \times n-j$ submatrix of $T_0^{(n)}$. Then 
$$\Vert T^{(n)}_{-j} \oplus T^{(n)}_{n-j} \Vert \leq \Vert T^{(n)}_0 \Vert \ \ \text{ for } j=1,2,\ldots,n-1.$$ \end{lemma}

By Theorems \ref{iso-T} and \ref{iso-C} together with the simple observation that $\mathcal B_0 =\mathcal V_0$, $\mathcal B_j = \mathcal V_{-j}\oplus \mathcal V_{n-j}$, allows us to obtain that
\begin{align}\label{dif-phi} \Phi^{(n)}- \widetilde{\Phi}^{(n)}\cong T^{(n)}_0-C^{(n)} \oplus  \bigoplus_{j=1}^{n-1}\Big (  T^{(n)}_{-j}\oplus T^{(n)}_{n-j} - C^{(n)}\Big )\ \ \text{ for each }n\geq 1.  \end{align}

\begin{lemma}\label{suma-HS} The Hilbert-Schmidt norm of the blocks in the last equation satisfies:
\begin{enumerate} \item If  $j=1,\ldots, n-1,$ then \begin{align*}  \left| T^{(n)}_{-j}\oplus T^{(n)}_{n-j} - C^{(n)} \right |^2 &=2\sum_{l=0}^{j-1}\sum_{k=1}^{n-j}(t_{-n+k+l}+t_{k+l})^2 +\sum_{k=j}^{n-1}(k-j)(t_k^2+t_{-k}^2)\\
&+\sum_{k=1}^{j-1}(j-k)(a_{n-k}^2+a_{-(n-k)}^2). \end{align*}
 \item \begin{align*} \sum_{j=1}^{n-1} \left| T^{(n)}_{-j}\oplus T^{(n)}_{n-j} - C^{(n)} \right |^2 &=2\sum_{j=1}^{n-1}j(n-j)(t_{-j}+t_{n-j})^2 + \sum_{j=2}^{n-1} j (j-1)(t_j^2+t_{-j}^2).\end{align*}
\item As a consequence
\begin{align} \label{cp-lim} \frac{1}{n^2}\sum_{j=1}^{n-1} \left| T^{(n)}_{-j}\oplus T^{(n)}_{n-j} - C^{(n)} \right |^2  \underset{n}{\longrightarrow} 0\end{align} \end{enumerate}\end{lemma}
\begin{proof}  \textit{(1)}  and \textit{(2)} are straightforward by direct computations.

For \textit{(3)}, since  $\sum_{j=1}^\infty t_j^2 + t_{-j}^2 <  \left( \sum_{j=1}^\infty t_j + t_{-j} \right)^2< \infty, $  
for any $\epsilon>0$ there exists $N=N(\epsilon)$  such that  \begin{align} \sum_{j=N}^{\infty} t_{-j}+t_{j} \leq \frac{\sqrt{\epsilon}}{2}, \ \text{ and } \ \ \  \sum_{j=N}^{\infty} t_{-j}^2+t_{j}^2 \leq \frac{\epsilon}{2}. \end{align}

We will use \textit{2.} For the first sum notice that 
\begin{align*} \frac{2}{n^2}\sum_{j=1}^{n-1}j (n-j) (t_{-j}+t_{n-j})^2 &\leq 2\left (\frac{1}{n}\sum_{j=1}^{n-1} \sqrt{j (n-j)} (t_{-j}+t_{n-j}) \right)^2  \\ & < 2\left (\sum_{j=1}^{n-1}  (t_{-j}+t_{n-j}) \right)^2 =  2\left (\sum_{j=1}^{n-1} (t_{-j}+t_{j}) \right)^2   <  \infty. \end{align*}

 Now if $n > N$ then 
\begin{align*}\frac{1}{n}\sum_{j=1}^{n-1} \sqrt{j (n-j)} (t_{-j}+t_{j})&=\frac{1}{n}\sum_{j=1}^{N-1} \sqrt{j (n-j)} (t_{-j}+t_{j})+\frac{1}{n}\sum_{j=N}^{n-1} \sqrt{j (n-j)} (t_{-j}+t_{j}) \\
&< \frac{1}{\sqrt{n} }\sum_{j=1}^{N-1} \sqrt{j} (t_{-j}+t_{j})+\frac{1}{n}\sum_{j=N}^{\infty} \sqrt{j (n-j)} (t_{-j}+t_{j}) \\
&\leq \frac{1}{\sqrt{n} }\sum_{j=1}^{N-1} \sqrt{j} (t_{-j}+t_{j})+ \frac{\sqrt{\epsilon}}{2} \underset{n}{\longrightarrow} \frac{\sqrt{\epsilon}}{2}.	 \end{align*}

While in a similar manner for the second term if $n > N$ we have
\begin{align*}\frac{1}{n^2}\sum_{j=2}^{n-1} j (j-1)(t_j^2+t_{-j}^2)&= \frac{1}{n^2}\sum_{j=2}^{N-1} j (j-1)(t_j^2+t_{-j}^2) +\frac{1}{n^2}\sum_{j=N}^{n-1} j (j-1)(t_j^2+t_{-j}^2) \\  
&< \frac{1}{n^2}\sum_{j=2}^{N-1} j (j-1)(t_j^2+t_{-j}^2) +\sum_{j=N}^{\infty} (t_j^2+t_{-j}^2) \\
&\leq  \frac{1}{n^2}\sum_{j=2}^{N-1} j (j-1)(t_j^2+t_{-j}^2)  + \frac{\epsilon}{2} \underset{n}{\longrightarrow} \frac{\epsilon}{2}.	\end{align*}

This shows that $\displaystyle \lim_{n\to \infty} \frac{1}{n^2}\sum_{j=1}^{n-1} \left| T^{(n)}_{-j}\oplus T^{(n)}_{n-j} - C^{(n)} \right |^2 =0$.

\end{proof}
We are now in position to prove the main theorem of this section.
\begin{theorem} \label{AE-CP}If $t=\{t_m\}_m\in \ell^1(\mathbb Z)$, the sequence of Toeplitz CP maps $\{\Phi^{(n)}\}_n$ is assymptotically equivalent to be a sequence of CP Circulant maps $\{\widetilde{\Phi}^{(n)}\}_n$, i.e.,  $\Phi^{(n)}\thicksim \widetilde{\Phi}^{(n)}$. \end{theorem}

\begin{proof} By combining the previous lemmas and the fact that by construction $T_0^{(n)} \thicksim C^{(n)}$, the uniform boundness follows since for any $ n\geq 1$,
\begin{equation} \label{norm-CP} \Big \Vert \Phi^{(n)} \Big \Vert= \Big \Vert T_0^{(n)}\oplus\bigoplus_{j=1}^{n-1} T^{(n)}_{-j}\oplus T^{(n)}_{n-j} \Big \Vert =\max_{1\leq j\leq n-1} \{\Vert T^{(n)}_0 \Vert, \Vert T^{(n)}_{-j}\oplus T^{(n)}_{n-j}\Vert \}=\Vert T^{(n)}_0\Vert < M,  \end{equation}
\[\Big \Vert \widetilde{\Phi}^{(n)} \Big \Vert= \Big \Vert [\widetilde{\Phi}^{(n)}]_\mathcal B \Big \Vert =\max_{0\leq j\leq n-1} \Vert C^{(n)}\Vert < M.\]

For the Hilbert-Schmidt norm, using (\ref{dif-phi}) and the orthogonality of different blocks we have
 \begin{align*} \Big| \Phi^{(n)}- \widetilde{\Phi}^{(n)}\Big|^2_{n^2} &= \frac{1}{n}\left|T^{(n)}_0-C^{(n)}\right|^2_n+ \frac{1}{n^2}\sum_{j=1}^{n-1} \left| T^{(n)}_{-j}\oplus T^{(n)}_{n-j} - C^{(n)} \right |^2.\end{align*}

The assymptotic equivalence implies that  $	\displaystyle \left|T^{(n)}_0-C^{(n)}\right|_n \xrightarrow[\ \ n \ \ ]{}	0$, while the remaining term vanishes by Lemma \ref{suma-HS}.

We have thus showed that $\displaystyle \lim_{n\to \infty} \Big| \Phi^{(n)}- \widetilde{\Phi}^{(n)}\Big|^2_{n^2} =0 $. 
This concludes the proof. \end{proof}

\section{Toeplitz class of non degenerate WCLT QMS}\label{wclt-ae}

Let 
Recall the definition of weak coupling limit type (WCLT) generators contained in \cite{WCLT}. Let $h_n$ be an $n$-dimensional Hilbert space and $H_S=H_S^*,D \in \mathcal B(h_n)$ with spectral decomposition
\[
H_S= \sum_{\epsilon_m \in\sigma(H_S) }\epsilon_m P_{\epsilon_m}
\]
where $\epsilon_0<\epsilon_1\dots< \epsilon_k $ and $P_{\epsilon_m}$ is the spectral projection
corresponding to the eigenvalue $\epsilon_m$. The {\it Bohr frequencies}  are all the differences  $\omega=\epsilon_m-\epsilon_l$ with $\epsilon_m,\epsilon_l$ eigenvalues of $H_S$ and denote by
\[
B_+=\{\,\omega=\epsilon_m-\epsilon_l>0\,\mid\, \epsilon_m,\epsilon_l \in \hbox{\rm Sp}(H_S)\,\},
\]
the set of all strictly positive Bohr frequencies. Define for each $\omega \in B_+$ the operators
\begin{eqnarray*}
D_\omega&=\sum_{\omega=\epsilon_m-\epsilon_l} P_{\epsilon_l} D P_{\epsilon_m},\qquad \qquad  D^*_\omega=\sum_{\omega=\epsilon_m-\epsilon_l} P_{\epsilon_l} D^* P_{\epsilon_l} 
\end{eqnarray*}

The $H_S$-WCLT Markov generator
$\mathcal L^{(n)}$ on $\mathcal{B}(h_n)$ has the GKSL form given by:
\begin{align*}
\mathcal L^{(n)}(x)   =  \sum_{\omega\in {B}_+}\mathcal L^{(n)}_\omega(x),
\end{align*}\vspace{-.8cm}
\begin{align*}
\mathcal L^{(n)}_\omega(x)=\mathrm{i}[H_\omega,x]
&-\Gamma_{-\omega} \left(\frac{1}{2}\{D^*_\omega D_\omega,x\}
-D^*_\omega x D_\omega\right) \\ &- \Gamma_{+\omega}
\left(\frac{1}{2}\{D_\omega D^*_\omega,x\}-D_\omega x D^*_\omega\right).
\end{align*}

where $\Gamma_{\pm\omega} \geq 0$ and $H_\omega$ is a selfadjoint operator commuting with $H_S$,  $$H_{\omega}=\zeta_{-\omega}D^*_\omega D_\omega + \zeta_{+\omega}D_\omega D^*_\omega, \ \ \ \ \ \zeta_{-\omega},\zeta_{+\omega}\in \mathbb R.$$

WCLT generators are seen to be $\mathcal V$-invariant in the basis of $H_S$. Theorem 4.1 in Fagnola and Quezada\cite{FQ} gives a precise characterization of when a $\mathcal V$-invariant GKSL generator is WCLT. This theorem implies that WCLT QMS generators are Toeplitz for suitable choices of $H_S$ and $D$. Indeed one possible realization is to choose the harmonic osscilator, i.e., a non-degenerate $H_S$ with purely rank one spectral projections (so no multiplicities) system Hamiltonian with the form 
\[
H_S= \sum_{m_=0}^{n-1} m\omega \ P_m ,\ \ P_m=|e_m \rangle \langle e_m|, \ \ \omega>0,
\]
where $\{e_m\}_m$ is the basis of eigenvectors of $H_S$, and the interaction operator $D=\sum_{i=0}^{n-1} S^i$ is the sum of all powers of the left shift operator in $h_n$ with respect to the above basis. This basis will be the canonical basis in what follows. In this case the set of (strictly positive) Bohr frequencies is $B_+=\{ m\omega :  1 \leq m \leq n-1\}.$

Following the previous construction, to each strictly positive Bohr frequency $\omega_m=m\omega$ we associate the operators
\begin{align*}
D_m=D_{\omega_m}&=\sum_{j=0}^{{}n-(m+1)} P_i D P_{m+i},\qquad \qquad  D^*_m=D^*_{\omega_m}=\sum_{i=0}^{n-(m+1)} P_{m+i}D^* P_{i}.\
\end{align*}

Then each $\mathcal L_m=\mathcal L_{\omega_m}$ can be seen to have the form 
\begin{align*} \mathcal L_m(x)=i[H_{m},x]+\Gamma^-_{m}S^{*m}x S^m + \Gamma^+_{m} S^m x S^{*m}+\Delta^*_m x + x\Delta_m, \text{ where } \end{align*} 
\begin{align*} \Delta_m=\Delta^*_m=-\frac{1}{2}(\Gamma^-_{m} S^{*m}S^m + \Gamma^+_{m} S^m S^{*m}),\ \  H_m=H_m^*=\zeta^-_{m}S^{*m} S^m + \zeta^+_{m}S^m S^{*m}	. \end{align*}
 Omitting the notation of the dimension for now, the WCLT generator \begin{align}\label{WCLT}\mathcal L(x)=\sum_m \mathcal L_{m} (x)= \Phi(x) + G^*x + xG, \ \ G= \sum_m \Delta_m-iH_m \end{align} has the CP Toeplitz dissipative part $$\Phi(x)=\sum_{m=1}^{n-1} \Gamma^-_{m}S^{*m}x S^m + \Gamma^+_{m} S^m x S^{*m}$$ 
 with associated  Toeplitz matrix $T_0$ according to definition  \ref{def} taking $\Gamma^-_{m}=t_m$ and $\Gamma^+_{m}=t_{-m}$.

By the previous sections in order to fully characterize $\mathcal L$ on the diagonal subspaces $\mathcal V$ it just remains to study the non-dissipative part of the generator which for convenience we shall denote by $\Psi(x)=G^*x+xG$. It can be verified that $G$ is a multiplication operator (i.e., diagonal w.r.t. to the canonical basis).

\begin{definition} Given the scalars $\{\Gamma_m^+,\Gamma_m^-: m=1,\ldots,n-1  \}$, $\{\zeta_m^+,\zeta_m^-: m=1,\ldots,n-1  \}$, define
\begin{align*}  s&=\sum_{m=1}^{n-1} \Gamma^+_m +\Gamma^-_m ,\ \ \ \ \tilde{s}=\sum_{m=1}^{n-1} \zeta^+_m +\zeta^-_m,\\
s(k)&=\sum_{m=1}^{n-1-k}\Gamma^+_m+ \sum_{m=1}^k\Gamma^-_m \ \ \  \text{ for } k=0,1,\ldots,n-1, \\ 
\tilde{s}(k)&=\sum_{m=1}^{n-1-k}\zeta^+_m+ \sum_{m=1}^k\zeta^-_m, \ \ \  \text{ for } k=0,1,\ldots,n-1.  \end{align*} \end{definition}

\begin{lemma} \label{aux}The following relations hold 
\begin{enumerate} 
\item $ \displaystyle \frac{s(k)+s(k+j)}{2} < s\ \ $ for any $j=1,\ldots, n-1$ and $k=0,1,\ldots, n-1-j.$
\item $\displaystyle \sum_{j=1}^{n-1}\sum_{k=0}^{n-1-j} s- \frac{s(k)+s(k+j)}{2} =(n-1)	\sum_{m=1}^{n-1} m (\Gamma_m^+ + \Gamma_m	^-)$

\item $\displaystyle \sum_{k=0}^{n-1} \Big(\sum_{m=n-k}^{n-1} \Gamma^+_m + \sum_{m=k+1}^{n-1}\Gamma_m^- \Big)= \sum_{m=1}^{n-1}m( \Gamma_m^+ +\Gamma_m^-) $
 \end{enumerate}\end{lemma}

\begin{theorem} \label{WCLT-iso}\begin{enumerate} 

\item The restriction of $\Psi$ to the subspace $\mathcal V_l$, $-(n-1)\leq l \leq n-1$ is given by the $n-|l| \times n-|l|$ matrix 
\begin{align*}&G_0 =-\sum_{k=0}^{n-1} s(k) |e_k \rangle \langle e_k |\  \ \text{ for }   l=0,   \\ &G_l=\sum_{k=0}^{n-1-|l|}-\left( \frac{s(k)+s(k+|l|)}{2} + i\frac{l}{|l|} [\tilde{s}(k)-\tilde{s}(k+|l|)]\right)|e_k\rangle \langle e_k|, \text{ for } l\neq 0,\end{align*} 
that is, for $j=0,1,\ldots,n-1$ 
\begin{align*}  &v_{-j}\left(\Psi(|e_k \rangle \langle e_{k+j}|)\right)= G_{-j}e_k \ \ \text{ for all } k=0,1,\ldots,n-1-j, \\
&v_j\left(\Psi(|e_{k+j} \rangle \langle e_k|)\right)= G_{j}e_k \ \ \ \text{ for all } k=0,1,\ldots,n-1-j. \end{align*}

Therefore \begin{align} \label{G-theo} v_l(\Psi|_{\mathcal V_l}(x))=G_{|l|}v_l(x) \ \text{ for any} \ x\in\mathcal M_p(\mathbb C) \text{ and } l=-(n-1),\ldots,0,\ldots,n-1.\end{align}
\item The non-dissipative part $\Psi$ satisfies 
\begin{align*} \Psi \cong G_0 \oplus  \bigoplus_{j=1}^{n-1} G_{-j} \oplus G_{j}\cong G_0 \oplus \bigoplus_{j=1}^{n-1} G_{-j} \oplus G_{n-j}.\end{align*}

 \item $\Vert \Psi \Vert \leq s.$ 
  \end{enumerate} \end{theorem}

\begin{proof}For any $k=0,1,\ldots,n-1$ explicit computations show that 
\begin{align*} G e_k&= -\frac{1}{2} \sum_{m=1}^{n-1} \left(\Gamma^-_m S^{*m}S^m + \Gamma^+_m S^m S^{*m}\right)e_k -i\left(\zeta^-_m S^{*m} S^m + \zeta^+_m S^m S^{*m}\right)e_k\\& =-\frac{1}{2} \left(\sum_{m=1}^{n-1-k} \Gamma^+_m S^m e_{k+m} + \sum_{m=1}^k \Gamma_m^- S^{*m}e_{k-m} \right)-i\left(\sum_{m=1}^{n-1-k} \zeta^+_m S^m e_{k+m} + \sum_{m=1}^k \zeta^-_m S^{*m}e_{k-m} \right) \\
&=- \frac{1}{2}\left(\sum_{m=1}^{n-1-k}\Gamma^+_m+ \sum_{m=1}^k\Gamma^-_m \right)e_k-i\left(\sum_{m=1}^{n-1-k}\zeta^+_m+ \sum_{m=1}^k\zeta^-_m \right) e_k \\
&=- \frac{1}{2}\left( s- \sum_{m=n-1-k+1}^{n-1}\Gamma^+_m-\sum_{m=k+1}^{n-1}\Gamma^-_m \right) e_k - i\left( \tilde{s}- \sum_{m=n-1-k+1}^{n-1}\zeta^+_m-\sum_{m=k+1}^{n-1}\zeta^-_m	 \right) e_k \\
&=-\left(\frac{s(k) }{2}+i\tilde{s}(k) \right) e_k.
\end{align*}

Let $j\geq 0$. Since $\Psi(|e_k\rangle \langle e_{k+j}|)=|G^*e_k \rangle \langle e_{k+j}|+|e_k \rangle \langle G^*e_{k+j}|$,
using the above computation it is straightforward to see that $v_0(\Psi(|e_k \rangle \langle e_k |))=G_0 e_k$ holds. Thus
\begin{align*}v_{-j}\Big( \Psi(|e_k\rangle \langle e_{k+j}|) \Big)&=-\left( \frac{s(k+j)+s(k)}{2} +i [\tilde{s}(k+j)-\tilde{s}(k)]\right)  v_{-j}\Big(|e_k\rangle \langle e_{k+j}| \Big)\\&=G_{-j}e_k\end{align*}


Analogous computations show that $v_{j}\Big(\Psi(|e_{k+j}\rangle \langle e_k|)\Big)=G_j e_k,$ while by a standard linearity argument (\ref{G-theo})  and thus 2. follows.

Finally for the bound note that each $G_{-j},G_{j}$ is diagonal. Since $s(k)< s$, using 1. from Lemma \ref{aux} one concludes that
$$\Vert \Psi \Vert= \max_{1 \leq j \leq n-1}\{ \Vert  G_0 \Vert, \Vert G_{
-j} \oplus  G_{n-j} \Vert \}  \leq s. $$ \end{proof}

Putting together Theorem \ref{iso-T} and Theorem \ref{WCLT-iso} we have thus proved:

\begin{theorem}The WCLT generator  satisfies 
\begin{align*} \mathcal L&\cong T_0+ G_0 \oplus  \bigoplus_{l=1}^{n-1} T_{-l}+G_{-l} \oplus T_{l}+G_{l},  \\
&\cong T_0+ G_0 \oplus \bigoplus_{l=1}^{n-1} T_{-l}+G_{-l} \oplus T_{n-l}+ G_{n-l}.\end{align*}
\end{theorem}

\section{Assymptotic equivalence of WCLT and circulant generators}

Given two sequences in $\ell^1(\mathbb Z)$,  $\Gamma=\{\Gamma_m\}_m$, $\Gamma_m\geq 0$ and $\zeta=\{\zeta_m\}_m$, $\zeta_m\in \mathbb R$, consider the  sequence $\{\mathcal L^{(n)}_T\}_n$  of \textbf{WCLT} generators with C.P. Toeplitz dissipative part according to (\ref{WCLT}). So for each $n \geq 1$
\begin{align*} \mathcal L^{(n)}_T(x)&=\Phi^{(n)}(x) + G^{(n)*}x + xG^{(n)}, \\
\Phi^{(n)}(x)&=\sum_{m=1}^{n-1} \Gamma^-_mS^{*m}x S^m + \Gamma^+_m S^m x S^{*m},\\ 
G^{(n)}&=-\frac{1}{2} \Phi^{(n)}(\unit_n)-iH^{(n)}, \\
H^{(n)}&= \sum_{m=1}^{n-1} \zeta^-_mS^{*m} S^m + \zeta^+_m S^m S^{*m}. \end{align*}
where $$\Gamma_m=\left\{\begin{array}{lc}\Gamma_m^+  & \text{ if $m > 0$} \\ \Gamma_m^- &\text{ if $m<0$}\end{array}\right.,\ \text{  and  }\ \ \zeta=\left\{\begin{array}{lc}\zeta^+  & \text{ if $m > 0$} \\ \zeta^- &\text{ if $m<0$}\end{array}\right..$$

We remark that generators are \textit{adapted} to the sequences $\Gamma$, $\zeta$.

 Call the partial sums Let $s^{(n)}=\sum_{m=1}^{n-1} \Gamma_m^+ + \Gamma_m^-$ and $\tilde{s}^{(n)}=\sum_{m=1}^{n-1} |\zeta_m^+| + |\zeta_m^-|.$ 

Consider the sequence $\{\mathcal L^{(C)}_n\}_n$ of circulant GKSL generators given by
\begin{align*}\mathcal L^{(C)}_n(x)=  \widetilde{\Phi}^{(n)}(x) - s^{(n)}\unit_n. \end{align*}

 where $\widetilde{\Phi}^{(n)}$ is the C.P. circulant map defined in (\ref{CP-circ}) and $s^{(n)}=\sum_{j=0}^{n-1} c_j^{(n)}$.

The circulant matrix associated to $\mathcal L^{(C)}_n$ (see \cite{BQ})  is the matrix $Q^{(n)}=C^{(n)}-s^{(n)}\unit_n$.

%
It is well known that  $\Vert T_0^{(n)}\Vert \leq \Vert \Gamma \Vert $ and  $\Vert C^{(n)} \Vert \leq \Vert \Gamma \Vert.$

\begin{theorem}\label{AE-WCLT} Let    $\{\mathcal L^{(T)}_n\}_n$  be a sequence WCLT GKSL generators with C.P. Toeplitz dissipative part adapted to $\Gamma$ and $\zeta$.
If $\Gamma,\zeta\in \ell^1(\mathbb Z) $  then there exists  a sequence of circulant GKSL generators $\{\mathcal L^{(C)}_n\}_n$ such that   
$\mathcal L^{(T)}_n \thicksim\mathcal L^{(C)}_n$. \end{theorem}
\begin{proof} The sequences are uniformly bounded since for $n\geq 1$
\begin{align*} \Vert \mathcal L^{(T)}_n \Vert &= \Vert \Phi_n + G_n\Vert \leq M  +\Vert \Gamma \Vert	 
\\ \Vert \mathcal L^{(C)}_n \Vert &=\Vert Q^{(n)} \Vert =\Vert C^{(n)} - s^{(n)}\unit_n \Vert \leq 2\Vert \Gamma \Vert. \end{align*}

The representation theorems for the generators allow us to write
\begin{align}\label{HS-eq} \left| \mathcal L^{(n)}_T-\mathcal L^{(n)}_C \right|^2_{n^2}=\frac{1}{n^2}\left|T^{(n)}_0+G^{(n)}_0 - Q^{(n)}\right|^2 + \frac{1}{n^2}\sum_{j=1}^{n-1} \left|T^{(n)}_{-j}+G^{(n)}_{-j} \oplus T^{(n)}_{n-j}+G^{(n)}_{n-j}- Q^{(n)}\right|^2 .\end{align}

By the orthogonality of different blocks we have
\begin{align*}\left|T_0^{(n)}+G_0^{(n)}-Q^{(n)}\right|^2&=  \left|T_0^{(n)}-C^{(n)}\right|^2 + \left|G_0^{(n)} + s^{(n)}\unit_n \right|^2, \\
\left|T^{(n)}_{-j}+G^{(n)}_{-j} \oplus T^{(n)}_{n-j}+G^{(n)}_{n-j}- Q^{(n)}\right|^2& =\left | T^{(n)}_{-j}\oplus T^{(n)}_{n-j}- C^{(n)} \right|^2 + \left|G_{-j}\oplus G_{n-j}+	s^{(n)}\unit_n \right|^2.
\end{align*}

Thus
\begin{align}  \left| \mathcal L^{(n)}_T-\mathcal L^{(n)}_C \right|^2_{n^2} &= \nonumber \frac{1}{n}\left|T_0^{(n)}-C^{(n)}\right|_n^2 +  \frac{1}{n}\sum_{j=1}^{n-1}\left | T^{(n)}_{-j}\oplus T^{(n)}_{n-j}- C^{(n)} \right|^2_n  \\
& \label{ae-L}+ \frac{1}{n}\left|G_0^{(n)} + s^{(n)}\unit_n \right|^2_n +\frac{1}{n}\sum_{j=1}^{n-1}\left|G_{-j}\oplus G_{n-j}+	s^{(n)}\unit_n \right|^2_n.\end{align}

Recall that by construction $\left|T_0^{(n)}-C^{(n)}\right|_n \xrightarrow[\ \ n \ \ ]{}	0,$ while by Lemma \ref{suma-HS}
$$ \frac{1}{n}\sum_{j=1}^{n-1}\left | T^{(n)}_{-j}\oplus T^{(n)}_{n-j}- C^{(n)} \right|_n^2 \xrightarrow[\ \ n \ \ ]{}0.$$

It remains to show the sums in (\ref{ae-L}) vanish as well.
 
For the first sum take $\epsilon >0$  and let $N$ be such that $\displaystyle \sum_{k=N}^{\infty} \Gamma_k^++\Gamma_k^-  \leq \sqrt{\epsilon}.$

Using that $x^2+y^2 \leq (x+y)^2$ and 3. from Lemma \ref{aux}, if $n\geq N$ we have 

\begin{align*}\frac{1}{n}\left|G_0^{(n)} + s^{(n)}\unit_n \right|_n^2 &=\frac{1}{n^2}\sum_{k=0}^{n-1}\left|s^{(n)}- s^{(n)}(k)\right|^2 = \frac{1}{n^2} \sum_{k=0}^{n-1}\left(\sum_{m=n-k}^{n-1} \Gamma^+_m + \sum_{m=k+1}^{n-1}\Gamma_m^- \right)^2  \\
& \leq \frac{1}{n^2} \left(\sum_{k=0}^{n-1} \left[\sum_{m=n-k}^{n-1} \Gamma^+_m + \sum_{m=k+1}^{n-1}\Gamma_m^- \right]\right)^2  \\
&= \left(\frac{1}{n}\sum_{k=1}^{N-1}k( \Gamma_k^+ +\Gamma_k^- )+\frac{1}{n}\sum_{k=N}^{n-1}k( \Gamma_k^+ +\Gamma_k^- )\right)^2  \\
&\leq \left(\frac{1}{n}\sum_{k=1}^{N-1}k( \Gamma_k^+ +\Gamma_k^- )+\sum_{k=N}^{\infty} ( \Gamma_k^+ +\Gamma_k^- )\right)^2 \\
&\leq \left(\frac{1}{n}\sum_{k=1}^{N-1}k( \Gamma_k^+ +\Gamma_k^- )+ \sqrt{\epsilon}\right)^2 \xrightarrow[\ \ n \ \ ]{} \epsilon  \end{align*}

The remaining term is in general a sum  of  the square modulus of complex numbers. For convenience call 
\begin{align*} \mathcal R(k,k+j)=\frac{s(k)+s(k+j)}{2}, \ \ \ \ \ \ \mathcal I(k,k+j)= \tilde{s}(k)-\tilde{s}(k+j). \end{align*}

A change of variables in the second allows us to simplify each term of the last sum in (\ref{ae-L}) 
\begin{align}\left|G_{-j}\oplus G_{n-j}+	s^{(n)}\unit_n \right|^2&=\nonumber \sum_{k=0}^{n-1-j} \left| \left( \mathcal R(k,k+j)-s^{(n)} \right)  +i \left(\frac{j}{|j|}  \mathcal I(k,k+j)\right) \right|^2 \\ \nonumber
 & \nonumber+ \sum_{k=0}^{j-1} \left| \left( \mathcal R(k,k+n-j)-s^{(n)} \right)  +i \left(\frac{n-j}{|n-j|} \mathcal I(k,k+n-j) \right) \right|^2 \\
=&\nonumber 2\sum_{k=0}^{n-1-j} \left| \left( \mathcal R(k,k+j)-s^{(n)} \right)  +i \left(\frac{j}{|j|} \mathcal I(k,k+j) \right) \right|^2 \\
=& \label{complex-sum} 2\sum_{k=0}^{n-1-j} \Big( \mathcal R(k,k+j)-s^{(n)} \Big)^2  +2\sum_{k=0}^{n-1-j} \Big( j\  \mathcal I(k,k+j) \Big)^2   \end{align}

For $\epsilon >0$ let $N$ be such that  $ \displaystyle \sum_{m=N}^\infty \Gamma_m^+ +\Gamma_m^- \leq \frac{\epsilon}{4 \Vert \Gamma \Vert}.$

Using $x^2=a^2+2a(x-a)+(x-a)^2$ and \textit{(1)},\textit{(2)} of Lemma \ref{aux}, if $n>N$ then the first sum of (\ref{complex-sum}) is
\begin{align*}  &\frac{2}{n^2} \sum_{j=1}^{n-1}\sum_{k=0}^{n-1-j} \Big( \mathcal R(k,k+j)-s^{(n)}\Big)^2  \\&=\frac{2}{n^2} \sum_{j=1}^{n-1}\sum_{k=0}^{n-1-j}\left(\mathcal R(k,k+j)^2 - \Big(s^{(n)}\Big)^2-2s^{(n)} \Big(\mathcal R(k,k+j)-s^{(n)}\Big) \right) \\
&= \frac{2}{n^2} \sum_{j=1}^{n-1}\sum_{k=0}^{n-1-j} \mathcal R(k,k+j)^2 -\frac{(n-1)}{n}\Big( s^{(n)}\Big)^2 +\frac{4	s^{(n)}}{n^2}\sum_{j=1}^{n-1}\sum_{k=0}^{n-1-j}\left[s^{(n)}-\frac{s^{(n)}(k)+s^{(n)}(k+j)}{2}\right] \\
&<\frac{2}{n^2} \sum_{j=1}^{n-1}\sum_{k=0}^{n-1-j} \Big( s^{(n)}\Big)^2 -\frac{(n-1)}{n}\Big(s^{(n)}\Big)^2  + \frac{4(n-1)}{n^2}s^{(n)}\sum_{m=1}^{n-1}m(\Gamma_m^+ +\Gamma_m^-) \\
&\leq \frac{4(n-1)}{n^2}s^{(n)}\sum_{m=1}^{N-1}m(\Gamma_m^+ +\Gamma_m^-)+\frac{4(n-1)}{n^2}s^{(n)}\sum_{m=N}^{\infty}m(\Gamma_m^+ +\Gamma_m^-) \\
& <\frac{4}{n}s^{(n)}\sum_{m=1}^{N-1}m(\Gamma_m^+ +\Gamma_m^-)+4s^{(n)}\sum_{m=N}^{\infty}(\Gamma_m^+ +\Gamma_m^-)  \xrightarrow[\ \ n \ \ ]{}	\epsilon,  \end{align*}
where the convergence $s^{(n)}\xrightarrow[\ \ n \ \ ]{}	\Vert \Gamma \Vert$ was used. 

On the other hand since $\tilde{s}^{(n)}\xrightarrow[\ \ n \ \ ]{}	\Vert \zeta \Vert$, the second sum of (\ref{complex-sum}) satisfies

\begin{align*}\frac{2}{n^2}\sum_{j=1}^{n-1}\sum_{k=0}^{n-1-j} \Big( j\  \mathcal I(k,k+j) \Big)^2 &\leq 2 \Big(\frac{1}{n}\sum_{j=1}^{n-1}\sum_{k=0}^{n-1-j}  \sqrt{j}\ \mathcal I(k,k+j) \Big)^2 \\ 
&\leq    2 \Big(\frac{1}{\sqrt{n}}\sum_{j=1}^{n-1}\sum_{k=0}^{n-1-j} \mathcal I(k,k+j) \Big)^2  \\
&= 2 \Big(\frac{1}{\sqrt{n}} \Big( \sum_{l=1}^{n-1}\zeta^+_l -\sum_{l=1}^{n-1}\zeta^-_l \Big) \Big)^2 \\
&\leq 2 \Big(\frac{1}{\sqrt{n}} \tilde{s}^{(n)}  \Big)^2  \xrightarrow[\ \ n \ \ ]{}	0\end{align*}

This shows that
$$\lim_n \left| \mathcal L^{(n)}_T-\mathcal L^{(n)}_C \right|_{n^2} =0. $$ 
\end{proof}

\section*{Acknowledgement}

\end{document}